\documentclass[11pt]{article} % use larger type; default would be 10pt

\usepackage[utf8]{inputenc} % set input encoding (not needed with XeLaTeX)

\usepackage{amssymb, amsthm, amsmath} 
\usepackage{geometry} % to change the page dimensions
\usepackage{graphicx} % support the \includegraphics command and options

\usepackage{booktabs} % for much better looking tables
\usepackage{array} % for better arrays (eg matrices) in maths
\usepackage{paralist} % very flexible & customisable lists (eg. enumerate/itemize, etc.)
\usepackage{verbatim} % adds environment for commenting out blocks of text & for better verbatim
\usepackage{subfig} % make it possible to include more than one captioned figure/table in a single float
\usepackage{fancyhdr} % This should be set AFTER setting up the page geometry
\usepackage{sectsty}
\allsectionsfont{\sffamily\mdseries\upshape} % (See the fntguide.pdf for font help)
\usepackage[nottoc,notlof,notlot]{tocbibind} % Put the bibliography in the ToC
\usepackage[titles,subfigure]{tocloft} % Alter the style of the Table of Contents

\newcommand{\F}{\mathbb F}
\newcommand{\Tr}{\mathbf{Tr}_{\F_{q^{m}} \setminus \F_{q}}}
\newcommand{\Trt}{\mathbf{Tr}_{\F_{q^3} \setminus \F_{q}}}
\newtheorem{definition}{Definition}
\newtheorem{proposition}{Proposition}
\newtheorem{theorem}{Theorem}
\newtheorem{lemma}{Lemma}
\newtheorem{corollary}{Corollary}
\title{Improving the minimum distance bound of Trace Goppa codes}
\author{Isabel Byrne, Natalie Dodson, Ryan Lynch,  Eric Pab\'on and Fernando Pi\~nero}
\begin{document}
\maketitle

\begin{abstract}
  In this article we prove that a class of Goppa codes whose Goppa polynomial is of the form $g(x) = \Tr(x)$ (i.e. $g(x)$ is a trace polynomial from a field extension of degree $m \geq 3$) has a better minimum distance than what the Goppa bound $d \geq 2deg(g(x))+1$ implies. Our improvement is based on finding another Goppa polynomial $h$ such that $C(L,g) = C(M, h)$ but $deg(h) > deg(g)$. This is a significant improvement over Trace Goppa codes over quadratic field extensions (i.e. the case $m = 2$), as the Goppa bound for the quadratic case is sharp.  
\end{abstract}

\section{Introduction}\label{sec1}

Binary Goppa codes are one of the fundamental linear code constructions in Coding Theory. Binary Goppa codes have been extensively studied since their introduction by V.D. Goppa in \cite{Goppa-70}. Their rich algebraic structure, and good decoding capabilities make binary Goppa codes suitable candidates for cryptography applications. There are also Best Known Linear Codes constructions realized by binary Goppa codes.

Throughout this article we shall assume $q$ is a prime, $q = p^s$ for some natural number $s$, and $m \geq 3.$ We focus on binary Goppa codes where the defining polynomial $g(x)$ is of the form $g(x) = \Tr(x)$, or that is $$g(x)  = x +x^q +x^{q^2} + \cdots + x^{q^{m-1}}.$$

\begin{definition}\cite{Goppa-70}
Let $q = p^s$ be a prime power. Let $q^m$ be a power of $q$. Suppose $L = \{ \alpha_1, \alpha_2, \ldots, \alpha_n\} \subseteq \F_{q^m}$. Let $g(x)$ be a univariate polynomial of degree $t$ such that $g(\alpha_i) \neq 0 $, for $\alpha_i \in L$. The \emph{$p$--ary Goppa code} is defined as $$C(L, g) := \{ (c_1, c_2, \ldots, c_n) \in \F_{p}^n: \sum\limits_{i=1}^n \frac{c_i}{x-\alpha_i} \equiv 0 \mod g(x)     \}.$$
\end{definition}

Our dimension bound looks slightly different $dim C(L,g) \geq  n - mst$ than the classical dimension bound, $dim C(L,g) \geq  n - mt$. This is because in the classical definition, the set $L$ is defined over $\F_q$ where $q = p^m$ whereas in our case, the set $L$ is defined in $\F_{q^m}$ where $q^m = p^{ms}$. The reason for the different bound is that our polynomial $g(x) = \Trt$ is also defined taking into consideration the subfield $\F_q$ of $\F_{q^m}$. Thus there are two subfields to consider: the subfield in which $g(x)$ takes values and the subfield over $C(L,g)$ is defined. We shall take $L\subseteq \F_{q^m}$, $g$ will take values in the subfield $\F_q$ and $C(L,g)$ will be defined in $\F_p$. Goppa codes may also be defined as subcodes over any subfield of $\F_{q^m}$. However, Goppa codes over the prime subfield (and in particular the binary subfield) remain the most interesting. Although our results hold for any subfield $\F_{q_0} \subseteq \F_{q^m}$ for the sake of simplicity in this article Goppa codes are defined over the prime field $\F_p$. The following bound on the dimension of binary Goppa codes is well known:

\begin{proposition}\cite{Goppa-70}
Let $q = p^s$ be a prime power. Let $q^m$ be a power of $q$. Let $L = \{ \alpha_1, \alpha_2, \ldots, \alpha_n\} \subseteq \F_{q^m}$. Let $g(x)$ be a polynomial of degree $t$ such that $g(\alpha_i) \neq 0 $, for $\alpha_i \in L$.  Then the dimension  of $C(L, g)$ is at least $n- smt$ and the minimum distance of $C(L, g)$ is at least $t +1$.
\end{proposition}

One of the first improvements on the bounds of Binary Goppa codes was given by Goppa in \cite{Goppa-70}. This improvement is based on establishing that two different Goppa polynomials give the same binary Goppa codes. This allows to use one polynomial to bound the dimension of the code and another polynomial to bound the minimum distance of the code.

\begin{proposition}\cite{Goppa-70}
Let $q = 2^s$. Let $L = \{ \alpha_1, \alpha_2, \ldots, \alpha_n\}$ be a subset of $\F_{q^m}$. Let $g(x)$ be a squarefree polynomial of degree $t$ such that $g(\alpha_i) \neq 0 $, for $\alpha_i \in L$. Then the binary Goppa codes satisfy: $$C(L,g) = C(L, g^2).$$
\end{proposition}

This proposition improves the distance bound from $t+1$ to $2t+1$. The distance bound on the Goppa code $C(L, g)$ comes from the fact that the codewords of $C(L, g)$ satisfy certain special parity check equations. Sugiyama et. al generalize this equivalence between Goppa codes over arbitrary fields $\F_{q_0}$.

\begin{proposition}\cite{SKHN-76}
Let $q^m$ be a prime power. Let $\F_{q_0}$ be a subfield of $\F_{q^m}$. Let $L = \{ \alpha_1, \alpha_2, \ldots, \alpha_n\}$ be a subset of $\F_{q^m}$. Let $g(x)$ be a squarefree polynomial of degree $t$ such that $g(\alpha_i) \neq 0 $, for $\alpha_i \in L$. Then the Goppa codes defiend over $\F_{q_0}$ satisfy: $$C(L,g^{q_0-1}) = C(L, g^{q_0}).$$
\end{proposition}

\begin{definition}
Let $L = \{\alpha_1, \alpha_2, \ldots, \alpha_n \} \subseteq \F_{q^m}$ where $\# L = n$. Let $f(x) \in \F_{q^m}[X]$ be a polynomial. We define the evaluation map $ev$ as  $$ev_L: \F_{q^m}[X] \rightarrow \F_{q^m}^n, ev_L(f) = (f(\alpha_1), f(\alpha_2) \ldots, f(\alpha_n)).$$
\end{definition}

The map $ev$ is a linear map from the polynomial ring $\F_{q^m}[X]$ to the vector space $\F_{q^m}^n$. The kernel is $ker(ev) = \langle \prod\limits_{i=1}^n(X-\alpha_i) \rangle$. We find it more illustrating to work with $f(x) \in \F_{q^m}[X]$ to understand the parity check equations. From the definition of Goppa codes it follows that the parity check equations for $C(L, g)$ may also be written as evaluation maps $ev_L(f)$. We describe those parity check equations as follows.

\begin{proposition}\cite{Goppa-70}\label{prop:GoppaPCE}
Let $q$ be a prime power. Let $L = \{ \alpha_1, \alpha_2, \ldots, \alpha_n\} \subseteq \F_q$. Let $g(x)$ be a polynomial of degree $t$ such that $g(\alpha_i) \neq 0 $, for $\alpha_i \in L$.  Then any codeword $c = (c_1, c_2, \ldots, c_n) \in C(L,g)$ satisfies

$$  \sum\limits_{i=1}^n c_i \frac{\alpha_i^j}{g(\alpha_i)} =  c \cdot ev_L\left(\frac{X^j}{g(X)}\right) = 0 \makebox{ where } 0 \leq j \leq t - 1$$
\end{proposition}

Goppa codes belong to a class of codes known as Alternant Codes. Alternant codes are subfield subcodes of Generalized Reed--Solomon codes. One of the good things about Alternant codes is that one can get a bound on its minimum distance at follows:

\begin{proposition} 
Let $q = p^s$. Let $\alpha_1, \alpha_2, \ldots, \alpha_n$ be distinct elements in $\F_{q^m}$. Let $\mathbf{a} = (a_1, a_2, \ldots, a_n)$ be nonzero elements in $\F_{q^m}.$ Let $\delta$ be a positive integer. Let $C$ be a code of length $n$ over $\F_p$.

If $\sum\limits_{i=1}^n c_ia_i \alpha_i^j = 0$ for $0 \leq j \leq \delta-2$ and $(c_1, c_2, \cdots, c_n) \in C$ then the minimum distance of $C$ is at least $\delta$.\end{proposition}

Goppa codes are Alternant codes where $a_i = g(\alpha_i)^{-1}$. The classical Goppa distance bound comes from the consecutive powers from $j=0$ to $j= \deg(g)-1$. Our distance bound improvements come from finding more consecutive powers which are parity check equations for $C(L, \Tr(x))$.

Goppa codes are linear codes defined over a small field, $\F_p$. However the parity check equations describing the Goppa codes are defined over the larger field $\F_{q^m}$. For $x = (x_1, x_2, \ldots, x_n) \in \F_{q^m}^n$ denote by $$x^{(p^i)} = (x_1^{p_i}, x_2^{p_i}, \ldots, x_n^{o_i}) \in \F_{q^m}^n .$$ Note that if $c \in C(L, g)$ and $c \cdot ev_L\left(\frac{X^j}{g(X)}\right) = 0$ then $c^{(p^i)} \cdot ev_L\left(\frac{X^j}{g(X)}\right)^{(p^i)} = 0$. As $c \in \F_p$ it follows that $c^{(p^i)} = c$. Thus for each $p$--power, we get the additional parity check equations $c \cdot ev_L\left(\left(\frac{X^j}{g(X)}\right)^{p^i}\right) = 0$. As $q^m = p^{ms}$ and there are $ms$ different $p$--powers, this is how the dimension bound $\dim(C(L,g)) \geq n-mst$ is derived.  Recall that the trace function $\Tr(\alpha)$ takes values in the subfield $\F_q$ for any $\alpha \in \F_{q^m}$. This implies that $\Tr(\alpha)^q = \Tr(\alpha)$ and also that $ev_{L}\left(\frac{X^i}{g(X)^q} \right) = ev_{L}\left(\frac{X^i}{g(X)} \right) $. This fact will be important later when we prove that certain $p$--powers of evaluation vectors $ev_{L}\left(\frac{X^i}{g(X)^q} \right) $ and $ev_{M}\left(\frac{Y^j}{h(Y)^q} \right) $ are in the dual codes $C(L, g)^\perp$ or $C(M,h)^\perp$.

P. V\'eron has improved bounds on the dimension of Trace Goppa codes. In fact his bounds are sharp for $m=2$. S. Bezzatev and N. Shekhunova in \cite{BS-09} proved that the classical distance bound is sharp for $m = 2$. We improve the minimum distance for trace Goppa codes when $m \geq 3$ instead.

\section{Improving the Minimum Distance of Trace Goppa Codes}\label{sec2}

We improve the minimum distance bound of Trace Goppa codes with $g(x) = \Tr(x)$ by establishing that the Goppa code $C(L, \Tr(x))$ is equivalent to the Trace Goppa code $C(M, h(y))$ where $$h\left(y\right) = \Tr\left(y^{a}\right) \mod y^{q^m}-y$$ and $a = 1+q+\cdots + q^{m-2}$. If $b =  1+q+\cdots + q^{m-1}$ then $$h\left( y \right) = y^{b-1} + y^{b-q} + y^{b-q^2} + \cdots + y^{b-q^{m-1}}.$$ The degree of $h(y)$ is $b-1 = aq$.  Denote by $$L = \{ \alpha \in \F_{q^m} : \Tr(\alpha)  \neq 0 \}$$ and denote by $$M = \{\beta \in \F_{q^m} \ : \ h(\beta) \neq 0 \}.$$

Now we prove that $$\alpha \in L \makebox{ if and only if } \alpha^{-1} \in M. $$

\begin{lemma}
Let $\alpha \in \F_{q^m}$. Then $\alpha \in L$ if and only if $\alpha^{-1} \in M$ 

\end{lemma}
\begin{proof}
Let $\alpha \in L$. This implies that $\Tr(\alpha) \neq 0$. As $\Tr(0) = 0$, this implies $\alpha \neq 0.$ Therefore we divide $\Tr(\alpha)$ by $\alpha^{1+q+\cdots + q^{m-1}}$ and obtain $\frac{\Tr(\alpha)}{\alpha^{1+q+\cdots + q^{m-1}}} \neq 0$.

This implies $\frac{\alpha + \alpha^q + \cdots + \alpha^{q{^{m-2}}}}{\alpha^{1+q+\cdots + q^{m-1}}} \neq 0$. We rewrite the sum as $\frac{\alpha}{\alpha^{b}} + \frac{\alpha^q}{\alpha^{b}} + \cdots + \frac{\alpha^{q^{m-1}}}{\alpha^{b }}  \neq 0  $ where $b = 1+q+q^2+\cdots + q^{m-1}$. Collecting the different powers we obtain $\alpha^{-a} + \alpha^{-aq } + \cdots + \alpha^{-aq^{m-1}} \neq 0$ which implies $h(\alpha^{-1}) \neq 0$ and thus $\alpha^{-1} \in M$.

Now suppose that $\alpha^{-1} \in M$. Then $h(\alpha^{-1}) \neq 0$. From the definiton of $h$, it follows $$h(\alpha^{-1}) = \alpha^{-a} + \alpha^{-aq } + \cdots + \alpha^{-aq^{m-1}} \neq 0. $$  We rewrite the sum as $\frac{\alpha}{\alpha^{b}} + \frac{\alpha^q}{\alpha^{b}} + \cdots + \frac{\alpha^{q^{m-1}}}{\alpha^{b }}  \neq 0 .$ As $\alpha \neq 0$ it follows that $\alpha +\alpha^q +\cdots + \alpha^{q^{m-1}}\neq 0$ which implies $\alpha \in L$.\end{proof}

We've established a relation between elements which are not roots of $g$ and the elements which are not roots of $h$. Now we describe the relations amongst the parity check equations for $C(L, g)$ and $C(M,  h)$. As stated in Proposition \ref{prop:GoppaPCE} the parity check equations for $C(L, g)$ are vectors of the form $ev_L\left(\frac{X^i}{g(X)}\right)$ for $0 \leq i \leq q^{m-1}-1$ and the corresponding $ms$ $p$--powers for each of the vectors. Likewise the parity check equations for $C(M, h)$ are generated from parity check equations of the form $ev_M\left(\frac{Y^j}{h(Y)}\right)$ for $0 \leq i \leq q^{m-1}+q^{m-2} + \cdots + q$ and all their $p$--powers. Luckily, the trace polynomial $\Tr(X)$ takes values in the subfield $\F_q$. This implies there are relations among the different $p$--powers of the parity check equations. For example $ev_L\left(\frac{X^i}{g(X)}\right)^{(q) }= ev_{L}\left(\frac{X^{qi}}{g(X)}\right)$ and $ev_M\left(\frac{Y^j}{h(Y)}\right)^{(q) }= ev_{M}\left(\frac{Y^{qj}}{h(Y)}\right)$. P. V\'eron (\cite{Veron-01}) used these relations to improve the dimension bound from $n - msq^{m-1}$ to $n - (m-1)sq^{m-1}$. We find relations between the different parity check equations to improve the distance bounds. We begin with the following lemma.
\begin{lemma}

Suppose $0 \leq i \leq q^m-2$. Assume the $q$--ary expansion of $i = \sum\limits_{s=0}^{m-1}i_rq^r$ where $0 \leq i_r \leq q-1$. Then $qi \mod q^m-1 = \sum\limits_{r=0}^{m-1}i_{r-1}q^s $.
\end{lemma}

The parity check equations for $C(L,g)$ are $ev_L\left(\frac{X^i}{g(X)}\right)$ for $0 \leq i \leq q^{m-1}-1$ and their $p$--powers. We use the fact that $g(\alpha)^q = g(\alpha)$ for $\alpha \in \F_{q^m}$ to prove that  $ev_L\left(\frac{X^i}{g(X)}\right)$ for $q^{m-1} \leq i \leq q^{m-1}+q^{m-2}+\cdots+q$ may be obtained from a $q$--power of some $ev_L\left( \frac{X^j}{g(X)}\right)$ where $0 \leq j \leq q^{m-1}-1$.

\begin{lemma}
Let $q^{m-1} \leq i \leq q^{m-1}+q^{m-2}+\cdots+q$. Then $ev_L\left(\frac{X^i}{g(X)}\right)^{(q) } \in C(L, g)^\perp$.
\end{lemma}
\begin{proof}
Suppose $q^{m-1} \leq i \leq q^{m-1}+q^{m-2}+\cdots+q$. The $q$--ary expansion of $i$ is of the form $i = \sum\limits_{r=0}^{m-1}i_rq^r$ where at least one of the entries $i_r = 0$. Otherwise if each $i_r \geq 1$ then $i > q^{m-1}+q^{m-2}+\cdot+q$. If $i_r = 0$, then $i' = q^{m-1-r}i \mod q^m-1 < q^{m-1}$. Therefore $ev_L\left(\frac{X^i}{g(X)}\right)^{(q^ {m-1-r})} = ev_L\left(\frac{X^{i'}}{g(X)}\right) \in C(L, g)^\perp$.  \end{proof}

A similar technique proves the following lemma:

\begin{lemma}
$ev_M\left(\frac{X^{q^{m-1}+q^{m-2}+ \cdots +q}}{h(X)}\right) \in C(M, h)^\perp$.

\end{lemma}
\begin{proof}
Recall that $a = q^{m-2}+q^{m-3}+ \cdots +q+1$. Thus
$ev_M\left(\frac{X^{aq}}{h(X)}\right) = ev_M\left(\frac{X^{ a}}{h(X)}\right)^{(q)}$. 
As $ev_M\left(\frac{X^{ a}}{h(X)}\right)^{(q)}$ is the $q$--power of $ev_M\left(\frac{X^{ a}}{h(X)}\right)$ and $a < deg(h)$ it follows that the evaluation vector $ev_M\left(\frac{X^{a}}{h(X)}\right) \in C(M, h)^\perp$. Since $ev_M\left(\frac{X^{q^{aq}}}{h(X)}\right)$ is a $q$--power of $ev_M\left(\frac{X^{aq}}{h(X)}\right)$ and $C(M,h)^\perp$ contains all of its $p$--powers, $ev_M\left(\frac{X^{ aq}}{h(X)}\right) \in C(M, h)^\perp$ follows. \end{proof}
Note that both $C(L,g)$ and $C(M,h)$ have parity check equations of consecutive powers from $0$ to $aq$. We prove now that there is a change of variables which maps one set of parity check equations to the other.

\begin{lemma}\label{lem:TraceEquiv}
The codes $C(L, g)$ and $C(M,h)$ are equal.
\end{lemma}

\begin{proof}
We'll prove that the map $x \mapsto x^{-1}$ maps the parity check equations $ev_L\left(\frac{X^i}{g(X)}\right) \in C(L, g)^\perp$ for $0 \leq i \leq aq$ to the parity check equations $ev_M\left(\frac{Y^j}{h(Y)}\right) \in C(M, h)^\perp$ for $0 \leq j \leq aq$.  Let $ev_L\left(\frac{X^{i}}{g(X)}\right) \in C(L, g)^\perp$ for $0 \leq i \leq  aq$. Let $Y = X^{-1}$, then $$ev_L\left(\frac{X^{i}}{g(X)}\right) = ev_M\left(\frac{Y^{-i}}{g(Y^{-1})}\right) $$
We multiply both sides of the fraction by $Y^{aq}$ and obtain: $$ev_M\left(\frac{Y^{-i}}{g(Y^{-1})}\right) = ev_M\left(\frac{Y^{ aq}Y^{-i}}{Y^{ aq}g(Y^{-1})}\right) = ev_M\left(\frac{Y^{ aq -i}}{h(Y)}\right).$$
Setting $j = aq-i$ we obtain that $0 \leq j \leq aq$. Therefore $$ev_L\left(\frac{X^i}{g(X)}\right) = ev_M\left(\frac{Y^{ aq -i}}{h(Y)}\right) = ev_M\left(\frac{Y^{ j}}{h(Y)}\right) \in C(M,h)^\perp.$$

Therefore $C(M, h) \subseteq C(L, g)$. All steps are reversible, which implies equality.
\end{proof}

The equality between the Goppa codes $C(L, g)$ and $C(M,h)$ leads to improved bounds on the minimum distance $d(C(L, g))$. Lemma \ref{lem:TraceEquiv} implies that $C(M,h) = C(M, h^2)$ for binary Goppa code even though $h$ itself is not a square free polynomial. This leads to a significant improvement of the distance bound.

\begin{corollary}\label{cor:TraceGoppaEquals}
The codes $C(L, g^2)$ and $C(M, h^2)$ are equal.

\end{corollary}
\begin{proof}
%As in the proof of Lemma \ref{lem:TraceEquiv} we shall apply the map $x \mapsto x^{-1}$ to map the parity check equations $ev_L(\frac{X^i}{g(X)^2})$ for $0 \leq i \leq 2aq$ to the parity check equations $ev_M(\frac{Y^j}{h(Y)^2})$ for $0 \leq j \leq 2aq$.  Let $ev_L\left(\frac{X^{i}}{g(X)^2}\right) \in C(L, g)^\perp$ for $0 \leq i \leq  2aq$. Let $Y = X^{-1}$, then $$ev_L\left(\frac{X^{i}}{g(X)^2}\right) = ev_M\left(\frac{Y^{-i}}{g(Y^{-1})^2}\right) $$
%We multiply both sides of the fraction by $Y^{2aq}$ and obtain: $$ev_M\left(\frac{Y^{-i}}{g(Y^{-1})^2}\right) = ev_M\left(\frac{Y^{ 2aq}Y^{-i}}{Y^{ 2aq}g(Y^{-1})^2}\right) = ev_M\left(\frac{Y^{ 2aq -i}}{h(Y)^2}\right).$$
%Setting $j = 2aq-i$ we obtain that $0 \leq j \leq 2aq$. Therefore $$ev_L\left(\frac{X^i}{g(X)^2}\right) = ev_M\left(\frac{Y^{ 2aq -i}}{h(Y)^2}\right) = ev_M\left(\frac{Y^{ j}}{h(Y)^2}\right)$$ which is a parity check equation in $C(M,h^2)^\perp$.

%Therefore $C(L, g^2) \subseteq C(M, h^2)$. All steps are reversible, which implies equality.
The proof is the same is the one in Lemma \ref{lem:TraceEquiv} using $g^2$ and $h^2$ instead of $g$ and $h$ and using $Y^{2aq}$ intead of $Y^{aq}$. All other steps are equal. \end{proof}

\begin{corollary}

Let $q = 2^s$. Suppose $C(L, g)$ is a binary Goppa code. The minimum distance of $C(L,g)$ is at least $2(q^{m-1}+q^{m-2}+ \cdots +q) + 1$. \end{corollary}
\begin{proof}
Let $q$ be a power of $2$. Note that $g = \Tr(x)$ is a squarefree polynomial. Therefore the binary Goppa code $C(L, g)$ is equal to $C(L, g^2)$. As $C(L,g^2) = C(M,h^2)$ and $deg(h) = 2(q^{m-1}+q^{m-2}+ \cdots +q)$, the bound follows. \end{proof}

This lemma states a sufficient condition for a funcion of the form $ev_L(X^i)$ to be in $C(L, g)^\perp$.

 \begin{lemma}\label{lem:Xpow}

 Let $g$ be a monic polynomial of degree $t$. Let $i \geq t$. Suppose that $ev_{L}\left( \frac{X^{i'}}{g(X)} \right) \in C(L, g)^\perp$ for all $0 \leq i' <  i+t$. Then $$ev_L\left(\frac{X^{i+t}}{g(X)}\right) \in C(L, g)^\perp \makebox{  if and only if } ev_L(X^{i}) \in C(L, g)^\perp.$$
\end{lemma}
 \begin{proof}
 Suppose that $ev_{L}\left( \frac{X^{i'}}{g(X)} \right) \in C(L, g)^\perp$ for all $0 \leq i' <  i+t$. 
 
 Note that $$ev_{L}\left(X^i\right) - ev_{L}\left( \frac{X^{i+t}}{g(X)}\right) = ev_{L}\left(X^i - \frac{ X^{i+t}}{g(X)}\right) = ev_{L}\left(\frac{  X^ig(X) - g_tX^{i+t}}{g(X)}\right).$$ As $g$ is a polynomial of degree $t$, all terms of $X^ig(X) - X^{i+t}$ have degree less than $i+t$. By the hypothesis of this Lemma; $ev_{L}\left(\frac{  X^ig(X) - X^{i+t}}{g(X)}\right) \in C(L,g)^\perp$. As the difference $$ev_{L}\left(X^i\right) - ev_{L}\left( \frac{X^{i+t}}{g(X)}\right) \in C(L, g)^\perp$$ it follows that $ev_L(\frac{X^{i+t}}{g}) \in C(L, g)^\perp$ if and only if $ev_L(X^{i}) \in C(L, g)^\perp$. \end{proof}

In the next Lemma we find more consecutive parity check equations of the form $ev_{M}(\frac{X^j}{h(X)^2})$ where $j$ is larger than $2deg(h)-1$.

\begin{lemma}\label{lem:morehighpowers}

Let $0 \leq j \leq 2aq + 2$. Then  $ev_M\left(\frac{Y^j}{h(Y)^2} \right) \in C(M,h^2)^\perp$.
\end{lemma}
\begin{proof} 
The definition of the binary Goppa code $C(M,h^2)$ establishes  $$ev_M\left(\frac{Y^j}{h(Y)^2} \right) \in C(M,h^2)^\perp, 0 \leq j \leq 2aq-1.$$ 

Now we consider larger $j$. If $2aq \leq j \leq 2aq + 1 $  then the $q$--base expression of $j$, $j = \sum\limits_{r=0}^{m-1} j_rq^r$ where $0 \leq j_r \leq q-1$ has at least one of the $j_r$ satisfy $j_r \leq 1$. Otherwise $j_r \geq 2$ implies $j \geq 2aq+2$. Therefore there is a number of the form $j' = jq^r \mod q^m-1$ which is smaller than $2q^{m-1}$. In this case $$ev_M\left(\frac{Y^j}{h(Y)^2} \right) =  ev_M\left(\frac{(Y^{j'})^{q^{m-r}}}{h(Y)^2} \right) = ev_M\left(\frac{Y^{j'}}{h(Y)^2} \right)^{(q^{m-r})}  \in C(M,h^2)^\perp.$$

For $ j  = 2aq + 2 $, note that $$\frac{h(Y)^2 Y}{h(Y)^2} = Y \makebox{ and } \frac{h(Y)^2 Y^{2}}{h(Y)^2} = Y^2.$$ However $h(Y)^2Y $ can be written as the sum of  $Y^{2aq+1}$ plus other smaller powers of $Y$. Likewise $h(Y)^2Y^2 $ can be written as the sum of  $Y^{2aq+2}$ plus other smaller powers of $Y$. Therefore $ev_M(Y) = ev_M\left(\frac{h(Y)^2 Y}{h(Y)^2}\right) \in C(M,h^2)^\perp$. As $C(M,h^2)$ is a binary Goppa code both $C(M,h^2)$ and $C(M, h^2)^\perp$ contain all $2$--powers of their codewords. Therefore $ev_M(Y)^2 = ev_M(Y^2) \in C(M, h^2)^\perp $. Lemma \ref{lem:Xpow} implies $ev_M\left(\frac{h(Y)^2 Y^{2}}{h(Y)^2}\right) \in C(M,h^2)^\perp$. As $h(Y)^2Y^2$ is the sum of $Y^{2aq+2}$ plus other lower powers of $Y$ it follows that the evaluation vector $ev_M\left(\frac{Y^{2aq+2}}{h(Y)^2} \right) \in C(M,h^2)^\perp$. \end{proof}
Now we shall prove that $C(M, h^2)^\perp$ contains additional consecutive parity check equations for negative powers of $X$, which further improves the minimum distance.

\begin{lemma}\label{lem:lowconsecutivepowers} Let $q$ be a power of $2$,
Let $0 < i < \frac{1+q+q^2+\cdots+q^{m-2}}{\frac{q}{2}-1}$. Then $$ev_M\left( \frac{Y^{-i}}{h(X)^2}\right) \in C(M, h^2)^\perp.$$
\end{lemma}

\begin{proof}

We shall prove that $$ev_M\left( \frac{Y^{-i}}{h(X)^2}\right)^{(\frac{q}{2})}   \in C(M, h^2)^\perp, 0 < i < \frac{1+q+q^2+\cdots+q^{m-2}}{\frac{q}{2}-1}.$$  We proceed by induction and start with $i = 1$. Since $1 < \frac{1+q+q^2+\cdots+q^{m-2}}{\frac{q}{2}-1}$. We take the $\frac{q}{2}$--power of $$ev_M\left( \frac{Y^{-1}}{h(X)^2}\right)^{\frac{q}{2}} = ev_M\left(\frac{Y^{-\frac{q}{2}}}{h(X)}\right) = ev_M\left(\frac{Y^{-\frac{q}{2}}h(X)}{h(X)^2}\right).$$

As the lowest degree term of $h(X)$ is $X^{1+q+q^2+\cdots+q^{m-2}}$, all terms of $Y^{-\frac{q}{2}}{h(X)}$ are between $0$ and $1+q+q^2+\cdots+q^{m-1}$. Thus $ev_M\left(\frac{Y^{-\frac{q}{2}}h(X)}{h(X)^2}\right) \in C(M,h^2)^\perp$.

Now let us suppose that for all $1 \leq k \leq j-1$, $ev_M\left( \frac{Y^{-k}}{h(X)^2}\right) \in C(M, h^2)^\perp$. Furthermore, suppose that $j < \frac{1+q+q^2+\cdots+q^{m-2}}{\frac{q}{2}-1} $.

Now take the $\frac{q}{2}$--power of $ev_M\left( \frac{Y^{-j}}{h(X)^2}\right)^{\frac{q}{2}} = ev_M\left(\frac{Y^{-j\frac{q}{2}}}{h(X)}\right) = ev_M\left(\frac{Y^{-j\frac{q}{2}}h(X)}{h(X)^2}\right).$ As $j < \frac{1+q+q^2+\cdots+q^{m-2}}{\frac{q}{2}-1} $, we have that $1+q+q^2 +\cdots + q^{m-2} -j\frac{q}{2} > -j$ which implies $ev_M\left( \frac{Y^{-j}}{h(X)^2}\right)^{\frac{q}{2}} \in C(M,h^2)^\perp.$
\end{proof}

\begin{theorem}
The minimum distance of the $p$--ary Goppa code $C(L, g)$ is at least $q^{m-1}+q^{m-2}+ \cdots + q+1$. If $p = 2$ then the minimum distane of the binary Goppa code is at least $2(q^{m-1}+q^{m-2}+ \cdots + q+1) + \lfloor\frac{1+q+q^2+\cdots+q^{m-2}}{\frac{q}{2}-1}\rfloor$.

\end{theorem}
\begin{proof} Lemma \ref{lem:TraceEquiv} established that $C(L, g)  = C(M, h)$. The degree of  $h$ is $q^{m-1}+q^{m-2}+\cdots + q$. Therefore $$d(C(L, g) ) \geq q^{m-1}+q^{m-2}+\cdots + q+1 .$$

In the case of binary Goppa codes, Corollary \ref{cor:TraceGoppaEquals} implies $C(L, g^2)  = C(M,h^2)$. Lemma \ref{lem:morehighpowers} implies $C(M, h^2)$ has $2(q^{m-1}+q^{m-2}+\cdots +q+1)+1$ consecutive powers of the form $ev_M\left(\frac{Y^{i}}{h(Y)^2} \right)$ as parity check equations. This implies $$d(C(M,h^2)) \geq 2(q^{m-1}+q^{m-2}+\cdots +q+1)+2.$$

Lemma \ref{lem:lowconsecutivepowers} implies there are an additional $\lfloor\frac{1+q+q^2+\cdots+q^{m-2}}{\frac{q}{2}-1}\rfloor$ parity check equations for $C(M,h^2)$ of the form $ev_M\left(\frac{Y^{-i}}{h(Y)^2} \right)$. As $C(L,g) = C(M, h^2)$ we obtain $$d(C(L,g)) \geq 2(q^{m-1}+q^{m-2}+ \cdots + q+1) + \lfloor\frac{1+q+q^2+\cdots+q^{m-2}}{\frac{q}{2}-1}\rfloor.$$ \end{proof}

\section{Further improvements for $m=3$}\label{sec3}
So far we have improved the bound on $d(C(L,g))$ in two ways. First we established $C(L, g) = C(L, g^2) = C(M, h^2)$ where $deg(h) > deg(g)$. This leads to an increase of the distance bound from $2q^{m-1}+1$ to $2(q+q^2+ \ldots +q^{m-2}+q^{m-1}) + 1$. Then we found additional consecutive powers $ev_{M}\left(\frac{X^j}{h(X)^2}\right)$ in the dual code $C(M,h^2)^\perp$ leading to $$d(C(L,g) \geq 2(1+q+q^2+ \ldots +q^{m-2}+q^{m-1}) + \lfloor\frac{1+q+q^2+\cdots+q^{m-2}}{\frac{q}{2}-1}\rfloor.$$ In the case $m=3$ and for most values of $q$ the latter improvement is $4$. The distance bound is in fact $d(C(L,\Trt(x))) \geq 2q^2+2q+6$. Now we improve this bound further for the case $m = 3$.

We shall prove that $ev_L\left( \frac{X^{2q^2+2q+3}}{\Trt^2} \right), ev_L\left( \frac{X^{2q^2+2q+5}}{\Trt^2} \right) \in C(L, \Trt^2)^\perp$. If true, then $C(L, \Trt^2)^\perp$ would contain $2q^2+2q+7$ consecutive powers, which implies the $d(C(L, \Trt)) \geq 2q^2+2q+8$. In Section \ref{sec2} we worked with $C(M,h^2)$ instead of $C(L,g^2)$ because that made it simpler to get the additional powers $ev_M\left(\frac{Y^j}{h(Y)^2}\right)$ where either $j \geq 2qa$ or $ev_M\left(\frac{Y^{-j}}{h(Y)^2}\right)$ where $i < \frac{1+q+q^2+\cdots+q^{m-2}}{\frac{q}{2}-1}$. We can find more consecutive powers but it involves more nuanced relations amongst the parity check functions and their $2$--powers. In this case it is much simpler to work with $g(x) = \Tr(x)$ directly.  We shall assume that $q$ is an even prime power such that $q \geq 8$. The Trace Goppa codes $C(L, \Tr)$ has dimension $0$ when $q = 2$ and has dimension $1$ when $q = 4$.

\begin{lemma}\label{lem:q3}

 $$ev_L\left(\frac{X^{2q^2+2q+3}}{\Trt^2}\right)^{(\frac{q}{2})} \in C(L, \Trt^2)^\perp $$
\end{lemma}
\begin{proof} Recall that in this case we are working over the field $\F_{q^3}$ which implies $ev_L(X^{q^3}) = ev_L(X)$. The evaluation vector 
$$ev_L\left(\frac{X^{2q^2+2q+3}}{\Trt^2}\right)^{(\frac{q}{2})} = ev_L\left(\frac{X^{(2q^2+2q+3)\frac{q}{2}}}{(\Trt^2)^{\frac{q}{2}}}\right) = ev_L\left(\frac{X^{q^3+q^2+q +\frac{q}{2}}}{\Trt^q}\right).$$ As we are working over $\F_{q^3}$ and $\Trt$ takes values in $\F_q$ we get 
$$ev_L\left(\frac{X^{q^3+q^2+q +\frac{q}{2}}}{\Trt^q}\right)  =ev_L\left(\frac{X^{q^2+q +\frac{q}{2}+1}}{\Trt}\right) $$ % =ev_L\left(\frac{X^{q^2+q +\frac{q}{2}+1}}{\Trt}\right) $$

We rewrite the quotient with $\Trt^2$ in the denominator. We obtain
  $$ev_L\left(\frac{X^{2q^2+2q+3}}{\Trt^2}\right)^{(\frac{q}{2})}=ev_L\left(\frac{X^{q^2+q +\frac{q}{2}+1}(X+X^q+X^{q^2})}{\Trt^2}\right) $$

Expanding the sum we obtain:

  $$ev_L\left(\frac{X^{2q^2+2q+3}}{\Trt^2}\right)^{(\frac{q}{2})}=ev_L\left(\frac{X^{q^2+q +\frac{q}{2}+2}  + X^{q^2+2q +\frac{q}{2}+1} + X^{2q^2+q +\frac{q}{2}+1}}{\Trt^2}\right) .$$ Each of the $X$ powers has a coefficient in its $q$--ary expansion equal to $1$, which implies each of the $X$ powers has a $q$ power which is less than $2q^2$. Therefore each of the $X$ powers is in $C(L, \Trt^2)^\perp$, which implies  $$ev_L\left(\frac{X^{2q^2+2q+3}}{\Trt^2}\right)^{(\frac{q}{2})} \in C(L, \Trt^2)^\perp $$ \end{proof}

\begin{lemma}\label{lem:q5} 
Let $q= 2^s$ where $s \geq 3$.
 $$ev_L\left(X^{2q+5}\right)^{(\frac{q}{2})} \in C(L, \Trt^2)^\perp $$
\end{lemma}
\begin{proof}

We shall write the vector $ev_L\left(X^{2q+5}\right)^{(\frac{q}{2})}$ as a combination of $2$--powers of evaluations of $ev_L\left(\frac{X^i}{\Trt^2}\right)$ where $0 \leq i < 2q^2$.

Since each power $X^i$ satisfies $0 \leq i < 2q^2$ and $q \geq 8$ the following evaluation vectors are elements of $C(L, \Trt^2)^\perp$: $$c_1 = ev_L\left(  \frac{X^{6q+3}}{\Trt^2}\right)^{(\frac{q}{2})}, c_2 = ev_L\left(  \frac{X^{q^2 + (\frac{q}{2}+1)q+3}}{\Trt^2}\right)^{(q^2)}, c_3 = ev_L\left(  \frac{X^{q^2 + \frac{q}{2}q+4}}{\Trt^2}\right)^{(q^2)}$$ $$c_4 = ev_L\left(  \frac{X^{q^2 + 2q + \frac{q}{2}+2}}{\Trt^2}\right), c_5  = ev_L\left(  \frac{X^{q^2 +4q+ \frac{q}{2}}}{\Trt^2}\right)$$

Taking the $2$ powers inside of the evaluation, we obtain:

$$c_1 = ev_L\left(  \frac{X^{3q^2+q + \frac{q}{2}}}{\Trt}\right), c_2 = ev_L\left(  \frac{X^{3q^2 + q + \frac{q}{2}+1} }{\Trt^2}\right), c_3 = ev_L\left(  \frac{X^{4q^2+q+\frac{q}{2}}}{\Trt^2}\right)$$ $$c_4 = ev_L\left(  \frac{X^{q^2 + 2q + \frac{q}{2}+2}}{\Trt^2}\right), c_5  = ev_L\left(  \frac{X^{q^2 +4q+ \frac{q}{2}}}{\Trt^2}\right)$$

%$ev_L\left(  \frac{X^{3q^2+q + \frac{q}{2}}}{\Trt}\right)^{\frac{q}{2}}$, $ev_L\left(  \frac{X^{3q^2 + q + \frac{q}{2}+1} }{\Trt^2}\right)$, $ev_L\left(  \frac{X^{4q^2+q+\frac{q}{2}}}{\Trt^2}\right)$, $ev_L\left(  \frac{X^{q^2 + 2q + \frac{q}{2}+2}}{\Trt^2}\right)$, $ev_L\left(  \frac{X^{q^2 +4q+ \frac{q}{2}}}{\Trt^2}\right)\in C(L, \Trt^2)$

%$$ev_L\left(  \frac{X^{3q^2+q + \frac{q}{2}}}{\Trt}\right) + ev_L\left(  \frac{X^{3q^2 + q + \frac{q}{2}+1}+X^{4q^2+q+\frac{q}{2}}}{\Trt^2}\right) + ev_L\left(  \frac{X^{q^2 + 2q + \frac{q}{2}+2}+X^{q^2 +4q+ \frac{q}{2}}}{\Trt^2}\right) $$

Now we rewrite $c_1$ the evaluation of a rational function with $\Trt^2$ in the denominator by multiplying both sides by $\Trt$.

% $$ + ev_L\left(  \frac{X^{3q^2 + q + \frac{q}{2}+1}+X^{4q^2+q+\frac{q}{2}}}{\Trt^2}\right) +  ev_L\left(  \frac{X^{q^2 + 2q + \frac{q}{2}+2}+X^{q^2 +4q+ \frac{q}{2}}}{\Trt^2}\right) $$

$$c_1 = ev_L\left(  \frac{X^{3q^2+q + \frac{q}{2}}(X+X^q+X^{q^2})}{\Trt^2}\right), c_2 = ev_L\left(  \frac{X^{3q^2 + q + \frac{q}{2}+1} }{\Trt^2}\right)$$ $$c_3 = ev_L\left(  \frac{X^{4q^2+q+\frac{q}{2}}}{\Trt^2}\right), c_4 = ev_L\left(  \frac{X^{q^2 + 2q + \frac{q}{2}+2}}{\Trt^2}\right), c_5  = ev_L\left(  \frac{X^{q^2 +4q+ \frac{q}{2}}}{\Trt^2}\right)$$

%$ev_L\left(  \frac{X^{3q^2+q + \frac{q}{2}}(X+X^q+X^{q^2})}{\Trt^2}\right)$, $ev_L\left(  \frac{X^{3q^2 + q + \frac{q}{2}+1} }{\Trt^2}\right)$, $ev_L\left(  \frac{X^{4q^2+q+\frac{q}{2}}}{\Trt^2}\right)$ $ev_L\left(  \frac{X^{q^2 + 2q + \frac{q}{2}+2}}{\Trt^2}\right)$, $ev_L\left(  \frac{X^{q^2 +4q+ \frac{q}{2}}}{\Trt^2}\right)\in C(L, \Trt^2)$

% We expand the sums as

% $$ev_L\left(\frac{X^{3q^2+q + \frac{q}{2}+1}}{\Trt^2}\right) + ev_L\left( \frac{X^{3q^2+2q + \frac{q}{2}}}{\Trt^2}\right) + ev_L\left( \frac{X^{4q^2+q + \frac{q}{2}}}{\Trt^2}\right) + ev_L\left(  \frac{X^{3q^2 + q + \frac{q}{2}+1}}{\Trt^2}\right) + ev_L\left(  \frac{X^{4q^2+q+\frac{q}{2}}}{\Trt^2}\right) + ev_L\left(  \frac{X^{q^2 + 2q + \frac{q}{2}+2}}{\Trt^2}\right)+ev_L\left(  \frac{X^{q^2 +4q+ \frac{q}{2}}}{\Trt^2}\right)
% $$

We expand $c_1$ as  $$ c_1 = ev_L\left(  \frac{X^{3q^2+q + \frac{q}{2}+1}}{\Trt^2}\right) +
ev_L\left(  \frac{X^{3q^2+2q + \frac{q}{2}}}{\Trt^2}\right)+ ev_L\left(  \frac{X^{4q^2+q + \frac{q}{2}} }{\Trt^2}\right)$$

As we are working in characteristic $2$ if we add $c_1+c_2+c_3+c_4+c_5$ we obtain $$ c_1+ c_2 +c_3+c_4 +c_5 = ev_L\left(  \frac{ X^{q^2 +2q +\frac{q}{2}}(X^2 + X^{2q} + X^{2q^2})   }{\Trt^2}\right) = ev_L\left(X^{2q+5}\right)^{(\frac{q}{2})} $$

%ev_L\left(  \frac{X^{3q^2+q + \frac{q}{2}+1}}{\Trt^2}\right) +
%ev_L\left(  \frac{X^{3q^2+2q + \frac{q}{2}}}{\Trt^2}\right)+ ev_L\left(  \frac{X^{4q^2+q + \frac{q}{2}} }{\Trt^2}\right)$, $ev_L\left(  \frac{X^{3q^2 + q + \frac{q}{2}+1} }{\Trt^2}\right)$, $ev_L\left(  \frac{X^{4q^2+q+\frac{q}{2}}}{\Trt^2}\right)$ $ev_L\left(  \frac{X^{q^2 + 2q + \frac{q}{2}+2}}{\Trt^2}\right)$, $ev_L\left(  \frac{X^{q^2 +4q+ \frac{q}{2}}}{\Trt^2}\right)\in C(L, \Trt^2)$
%After cancelling like terms we rewrite the sum as

%$$  $$

%We factor $X^{q^2 +2q +\frac{q}{2}}$ in the numerator and obtain:

%$$  $$

%Which equals $$ev_L\left(X^{2q+5}\right)^{\frac{q}{2}}.$$
As we have written $ev_L\left(X^{2q+5}\right)^{(\frac{q}{2})}$ as the sum of elements of $C(L, \Trt^2)^\perp$, the Lemma follows. \end{proof}

We finish this article proving $d(C(L, \Trt)) \geq 2q^2+2q+8$.

 %Note that $ev_L(\frac{g X^i}{g}) = ev_L(X^{i})$. If the lower order terms of $g(X) $ are in $C(L, g)$, namely $ev_L(\frac{ (X^t-g)X^i}{g}) \in C(L,g) $, since $ev_L(X^i) = ev_L(\frac{g X^i}{g}) - ev_L(\frac{X^{i+t}}{g})  = ev_L(\frac{ (X^t-g)X^i}{g})\in C(L, g)^\perp$ it follows  $ev_L(X^i) \in C(L,g)$ if and only if $ev_L(\frac{X^{i-t}}{g}) \in C(L, g)^\perp$
 %\end{proof}

\begin{theorem}

The minimum distance of $C(L, \Trt^2)$ is at least $2q^2+2q+8$.

\end{theorem}

\begin{proof}
%Note that if $C(L, \Trt^2)$ has $\delta$ consecutive powers of the form $ev_L\left(\frac{X^i}{\Trt}\right)$ then the distance of $C(L, \Trt^2)$ is at least $\delta+1$. Also note that P. Veron determined in \cite{Veron-01} that  the all ones codeword is a parity check equation for $C(L, \Trt^2)$. Thus the minimum distance of $C(L, \Trt)$ must be even.

By the definition of the Goppa code $$ev_L\left( \frac{X^i}{\Trt^2}  \right) \in C(L, \Trt^2)^\perp, 0 \leq i \leq 2q^2-1.$$

Let $2q^2 \leq i \leq 2q^2+2q+1$. Writing $i$ in base $q$ we obtain $i = 2q^2 + i_1q+i_0$, where  $0 \leq i_0, i_1 \leq q-1$. The condition $i \leq 2q^2+2q+1$ implies either $i_1 \leq 1 $ or $i_0 \leq 1$. Thus either $iq \mod q^3-1$ or $iq^2 \mod q^3-1$ are smaller than $2q^2$. Therefore either $$ev_L\left( \frac{X^i}{\Trt^2}  \right)^{(q)} \makebox{ or } ev_L\left( \frac{X^i}{\Trt^2}  \right)^{(q^2)} \in C(L, \Trt^2)^\perp, 2q^2 \leq i \leq 2q^2+2q+1.$$ Since $C(L, \Trt^2)^\perp$ is closed under powers of $2$ it follows that $$ev_L\left( \frac{X^i}{\Trt^2}  \right) \in C(L, \Trt^2)^\perp, 0 \leq i \leq 2q^2+2q+1.$$

Lemma \ref{lem:Xpow} implies $$ev_L(X^{2q+1}), ev_L(X^{q+1}), ev_L(X^{q+2}) , ev_L(X^{q+3}) \in C(L, \Trt^2)^\perp.$$    
Lemma \ref{lem:q3} implies $$ev_L(X^{2q+3}) \in C(L, \Trt^2)$$ and Lemma \ref{lem:q5} implies $$ev_L(X^{2q+5}) \in C(L, \Trt^2).$$ Because $C(L, \Trt^2)$ is closed under $2$--powers, then $$ev_L(X^{2q+2}), ev_L(X^{2q+4}), ev_L(X^{2q+6}) \in C(L, \Trt^2)^\perp.$$ Lemma \ref{lem:Xpow} implies that $$ev_L\left( \frac{X^i}{\Trt^2} \right) \in C(L, \Trt^2)^\perp, 0 \leq i \leq 2q^2+2q+6.$$ There are $2q^2+2q+7$ consecutive powers in $C(L, \Trt^2)^\perp$ which implies $$d(C(L,\Trt^2) \geq 2q^2+2q+8.$$\end{proof}

\section{Conclusion}

In this article we improve the minimum distance bound of Goppa codes of the class $C(L, \Tr(x))$ in two different ways. First, we have proven that a Goppa code of the form $C(L, \Tr(x))$ is equivalent to a Goppa code of the form $C(M, \Tr(x^{a}))$ where $M = \{\beta \ : \ \beta^{-1} \in L  \}$ and $a = 1+q+\cdots + q^{m-2}$. As the degree of $\Tr(x^{a})$  is much larger than the degree of $\Tr(x)$, the minimum distance bound is significantly improved from $q^{m-1}+1$ to $q^{m-1}+q^{m-2}+ \cdots + q+1$. In the binary case the improvement is from $2q^{m-1}+1$ to $2(q^{m-1}+q^{m-2}+ \cdots + q+1)+1$.  For the binary Goppa codes with $m = 3$ we have further improved the minimum distance bound by finding additional parity check equations corresponding to additional consecutive powers. In this case the distance bound increases from $2q^2+2q+4$ to $2q^2+2q+8$.

%%\input sn-article.bbl

%% Default %%
%%\input sn-sample-bib.tex%


\begin{thebibliography}{}
\bibitem{Goppa-70} V. D. Goppa, "A new class of linear error-correcting codes", Probl. Peredach. Inform., vol. 6, no. 3, pp. 24-30, Sept. 1970.
\bibitem{Berlekamp-73} E. Berlekamp, "Goppa codes," in IEEE Transactions on Information Theory, vol. 19, no. 5, pp. 590-592, September 1973, doi: 10.1109/TIT.1973.1055088.
\bibitem{Veron-98} P. Veron, "Goppa codes and trace operator," in IEEE Transactions on Information Theory, vol. 44, no. 1, pp. 290-294, Jan. 1998, doi: 10.1109/18.65104
\bibitem{Veron-01} P. Véron, True Dimension of Some Binary Quadratic Trace Goppa Codes. Designs, Codes and Cryptography 24, 81–97 (2001). https://doi.org/10.1023/A:1011281431366
\bibitem{Veron-05} P. Véron, Proof of Conjectures on the True Dimension of Some Binary Goppa Codes. Des Codes Crypt 36, 317–325 (2005). https://doi.org/10.1007/s10623-004-1722-4
\bibitem{COT-14} A. Couvreur, A. Otmani, J-P. Tillich, 'New identities relating wild Goppa codes" in Finite Fields and Their Applications, Volume 29, 2014, Pages 178-197, ISSN 1071-5797,
https://doi.org/10.1016/j.ffa.2014.04.007.
\bibitem{CC-98} A. Canteaut and F. Chabaud, "A new algorithm for finding minimum-weight words in a linear code: application to McEliece's cryptosystem and to narrow-sense BCH codes of length 511," in IEEE Transactions on Information Theory, vol. 44, no. 1, pp. 367-378, Jan. 1998, doi: 10.1109/18.651067.

\bibitem{BS-09} S. Bezzateev and N.A. Shekhunova, "Chain of Separable Binary Goppa Codes and Their Minimal Distance" in IEEE Transactions on Information Theory vol 54,  pp 5773 - 5778. doi: 10.1109/TIT.2008.2006442. 


\bibitem{SKHN-76} Y. Sugiyama, M. Kasahara, S. Hirasawa and T. Namekawa, "Further results on Goppa codes and their applications to constructing efficient binary codes," in IEEE Transactions on Information Theory, vol. 22, no. 5, pp. 518-526, September 1976, doi: 10.1109/TIT.1976.1055610.
\end{thebibliography}
\end{document}